\documentclass[a4paper,10pt]{article}
\usepackage{fullpage}
\usepackage[dvipdfmx]{graphicx}
\usepackage[T1]{fontenc}
\usepackage{lineno}
\usepackage{amsmath,amssymb,mathtools,amsthm}
\usepackage{url}
\usepackage{comment}
\usepackage{fancybox,ascmac}
\usepackage[subrefformat=parens]{subcaption}
\usepackage{color}
\usepackage{paralist}
\usepackage{comment}
\usepackage{authblk}
\usepackage{enumerate}
\usepackage{thm-restate}
\usepackage[colorlinks=true,linkcolor=blue,citecolor=blue]{hyperref}
\usepackage{url}
\usepackage{lineno}

\newtheorem{theorem}{Theorem}
\newtheorem{lemma}{Lemma}

\newtheorem{corollary}{Corollary}

\long\def\nop#1{}

\newcommand{\NP}{$\mathsf{NP}$}
\newcommand{\NPc}{$\mathsf{NP}$-complete}
\newcommand{\NPh}{$\mathsf{NP}$-hard}

\newcommand{\FS}{\rm forcing~set}
\newcommand{\AFS}{\rm anti-forcing~set}
\newcommand{\FN}{\rm forcing~number}

\newcommand{\FGPM}{\sc ForcingGPM} % Forcing Set for a given Perfect Matching
\newcommand{\AFGPM}{\sc Anti-ForcingGPM}
\newcommand{\FPM}{\sc ForcingPM} % Forcing Set for Perfect Matching
\newcommand{\AFPM}{\sc Anti-ForcingPM}
\begin{document}
\title{Hardness of Forcing Unique Perfect Matchings in Bipartite Graphs of Maximum Degree 3}
\author[1]{Ryoma Aoshima}
\author[1]{Takashi Horiyama}
\author[2]{Atsuki Nagao}
\author[1]{Fumiya Sakamoto}
\author[2]{Hibiki Sato}
\author[1]{Kazuhisa Seto}
\author[2]{Karin Umebayashi}
\affil[1]{Hokkaido University\\
\texttt{aoshima.ryoma.j3@elms.hokudai.ac.jp\\ \{horiyama,seto\}@ist.hokudai.ac.jp}}
\affil[2]{Ochanomizu University\\
\texttt{a-nagao@is.ocha.ac.jp}}
\date{}

\maketitle

\begin{abstract}
In a graph $G$, a set of edges $F$ is called a \emph{forcing set} if there exists a unique perfect matching $M$ such that $F \subseteq M$.
Similarly, a set of edges $A$ is called an \emph{anti-forcing set} if the graph with edge set $ E(G)\setminus A$ has a unique perfect matching.
It is known that, given a bipartite graph $G$ of maximum degree~$3$ and a perfect matching $M$, the problem of deciding whether there exists a forcing set of size at most $k$ for $M$ is {\NPc}.
Moreover, given a bipartite graph $G$ of maximum degree~$4$ and a perfect matching $M$, the problem of deciding whether there exists an anti-forcing set of size at most $k$ for $M$ is {\NPc}.
Furthermore, given a bipartite graph of maximum degree~$5$, the problem of deciding whether there exists a perfect matching $M$ that can be made unique by a forcing set of size at most $k$ is also {\NPc}.
In contrast, the computational complexity of deciding whether there exists a perfect matching $M$ that can be made unique by an anti-forcing set of size at most $k$ is not known, even for general graphs.
In this paper, we show that all of these problems remain {\NPc} even when restricted to bipartite graphs of maximum degree~$3$.
\end{abstract}

\section{Introduction}
A perfect matching $M$ of a graph $G$ is a set of edges such that every vertex of $G$ is incident to exactly one edge in $M$.
This concept is used in organic chemistry to describe Kekul\'{e} structures.
For example, the carbon skeleton of benzenoid hydrocarbons can be represented by hexagonal systems that admit perfect matchings, and thus Kekul\'{e} structures of hexagonal systems have been extensively studied.
Randi\'{c} and Klein~\cite{randic1985kekule} defined the innate degree of freedom as the minimum number of double bonds required to determine a Kekul\'{e} structure uniquely.

Based on the notion of the innate degree of freedom, Harary, Klein, and {\v{Z}}ivkovi{\v{c}}~\cite{harary1991graphical} defined the forcing number to be the minimum size of an edge set that uniquely determines a perfect matching.
Kleinerman~\cite{kleinerman2006bounds} proved that the {\FN} of both a $2m \times 2n$ rectangle and a torus is $mn$.
The complexity of computing an edge set that uniquely determines a perfect matching has been studied.
A forcing set $F$ for a perfect matching $M$ is a subset of edges in $M$ that uniquely determines $M$, i.e., there exists no other perfect matching including all edges in $F$.
Adams, Mahdian, and Mahmoodian~\cite{adams2004forced} proved that deciding whether there exists a {\FS} of size at most $k$ that uniquely determines a given perfect matching is {\NPc} even for bipartite graphs of maximum degree~$3$.
Afshani, Hatami, and Mahmoodian~\cite{afshani2009spectrum} showed that deciding whether there exists a perfect matching uniquely determined by a {\FS} of size at most $k$ is {\NPc} even for bipartite graphs of maximum degree~$5$.\footnote{Although the paper~\cite{afshani2009spectrum} claims {\NP}-completeness for bipartite graphs of maximum degree~$4$, the graph obtained after the reduction used to prove {\NP}-completeness actually has maximum degree~$5$.} 

There is also a dual concept to {\FS}.
Li~\cite{LI1997295} introduced the anti-forcing edge as an edge that is excluded from exactly one perfect matching and proved that hexagonal systems with an anti-forcing edge are truncated parallelograms.
Vuki{\v{e}}evi{\'c} and Trinajsti{\'c}~\cite{vukiveevic2007anti} introduced the notion of the anti-forcing number of a graph and defined it as the minimum cardinality of an edge subset whose deletion yields a graph having a unique perfect matching.
Deng~\cite{deng2007anti} computed the anti-forcing number for a single hexagonal chain, and later extended the result to double hexagonal chains~\cite{deng2008anti}.
An anti-forcing set for a perfect matching $M$ is a set of edges of a graph $G$ such that the graph obtained by removing all edges in the set from $G$ 
has the unique perfect matching $M$.
Deng and Zhang~\cite{deng2017anti} showed that deciding whether there exists an {\AFS} of size at most $k$ that uniquely determines a given perfect matching is {\NPc} even for bipartite graphs of maximum degree~$4$.
However, the complexity of deciding whether there exists a perfect matching uniquely determined by an {\AFS} of size at most $k$ has remained open.

\paragraph{\bf Our Contribution:}
We first propose a polynomial-time reduction from the problem of computing a {\FS} to the problem of computing an {\AFS} via edge subdivision in graphs.
Using this technique, we strengthen Deng and Zhang's result by proving that deciding whether there exists an {\AFS} of size at most $k$ that uniquely determines a given perfect matching remains {\NPc}, even for bipartite graphs of maximum degree~$3$.
Next, inspired by a gadget construction of Demaine, Karntikoon, and Pitimanaaree~\cite{DemaineKP25}, we introduce a gadget that reduces the maximum degree of a given bipartite graph from five to three.
This improves Afshani et al.'s result by showing that deciding whether there exists a perfect matching uniquely determined by a {\FS} of size at most $k$ is {\NPc} even for bipartite graphs of maximum degree~$3$.
Finally, combining these results, we resolve the open problem by proving that deciding whether there exists a perfect matching uniquely determined by an {\AFS} of size at most $k$ is {\NPc} even for bipartite graphs of maximum degree~$3$.

\paragraph{\bf Related Work:}
There exist many studies on the forcing set, the anti-forcing set, and the global forcing set for a perfect matching.
See the excellent survey paper~\cite{ZhangHLZ25}.
Furthermore, the forcing set has been widely studied for many graph problems: graph colorings~\cite{Harary07}, independent sets~\cite{Larson13}, dominating sets~\cite{Chartrand97}, independent dominating sets~\cite{Armada19}, minimum spanning trees, and shortest paths~\cite{GKOS26}.

\section{Preliminaries}
Let $G = (V(G), E(G))$ be a graph, where $V(G)$ is the vertex set and $E(G)$ is the edge set.
A graph $G$ is called a \emph{bipartite graph} if $V(G)$ can be partitioned into two sets $X$ and $Y$ such that for every edge, one endpoint is in $X$ and the other is in $Y$.
For each $v \in V(G)$, the \emph{degree} of $v$ is the number of edges incident to $v$.
A \emph{matching} is a set of edges $M \subseteq E(G)$ such that no two edges in $M$ share a common endpoint.
If $M$ covers all vertices, then $M$ is called a \emph{perfect matching}.
Let $\mathcal{M}(G)$ denote the set of perfect matchings of $G$.
For $M \in \mathcal{M}(G)$, a \emph{forcing set} of $M$ is an edge set $F \subseteq M$ such that $F$ is contained in no other perfect matchings in $\mathcal{M}(G)$.
Similarly, for $M \in \mathcal{M}(G)$, an \emph{anti-forcing set} of $M$ is an edge set $A \subseteq E(G) \setminus M$ such that $G'=(V(G), E(G)\setminus A)$ has only one perfect matching $M$.

The problem {\FGPM} is the problem of deciding, given a graph $G$, a perfect matching $M$, and a nonnegative integer $k$, whether there exists a {\FS} of $M$ of size at most $k$.
The problem {\FPM} is the problem of deciding, given a graph $G$ and a non-negative integer $k$, whether there exists a perfect matching $M$ uniquely determined by a {\FS} of size at most $k$.

\begin{lemma} \label{lem:FPM_in_NP}
    {\FGPM} and {\FPM} belong to {\NP}.
\end{lemma}

\begin{proof}
    As a certificate, we give a perfect matching $M$ and an edge set $F \subseteq M$ of size at most $k$.
    To verify that $F$ is a {\FS} of $M$, it suffices to check that the graph $G'$, obtained from $G$ by removing all edges in $F$ and all edges adjacent to edges in $F$, has a unique perfect matching.
    The transformation from $G$ to $G'$ can be done in polynomial time.
    It has been shown that the uniqueness of a perfect matching in $G'$ can be decided in $O(|E(G')| \log^4 |V(G')|)$ time~\cite{gabow1999unique}.
    Therefore, Lemma~\ref{lem:FPM_in_NP} holds.
\end{proof}

The problem {\AFGPM} is the problem of deciding, given a graph $G$, a perfect matching $M$, and a nonnegative integer $k$, whether there exists an {\AFS} of $M$ of size at most $k$.
The problem {\AFPM} is the problem of deciding, given a graph $G$ and a non-negative integer $k$, whether there exists a perfect matching $M$ uniquely determined by an {\AFS} of size at most $k$.

\begin{lemma} \label{lem:AFPM_in_NP}
    {\AFGPM} and {\AFPM} belong to {\NP}.
\end{lemma}

\begin{proof}
    A similar proof to that of Lemma~\ref{lem:FPM_in_NP} applies.
    As a certificate, we give a perfect matching $M$ and an edge set $A \subseteq E(G) \setminus M$ of size at most $k$.
    To verify that $A$ is an {\AFS} of $M$, it suffices to check that the graph $G'$, obtained from $G$ by removing all edges in $A$ has a unique perfect matching.
    The transformation from $G$ to $G'$ can be done in polynomial time.
    It has been shown that the uniqueness of a perfect matching in $G'$ can be decided in $O(|E(G')| \log^4 |V(G')|)$ time~\cite{gabow1999unique}.
    Therefore, Lemma~\ref{lem:AFPM_in_NP} holds.
\end{proof}

\section{On Forcing Set and Anti-Forcing Set for Perfect Matching}

\subsection{Subdivision of Edges}\label{sec:seg}
In this section, we provide a polynomial-time reduction from {\FGPM} and {\FPM} to {\AFGPM} and {\AFPM}, respectively, by subdividing edges.
Figure \ref{fig:subdivision} illustrates the subdivision of an edge $e$.
Let the subdivided edges be denoted by $e^u = \{u, u^e\}$, $e^{\textnormal{mid}} = \{u^e, v^e\}$, and $e^v = \{v^e, v\}$.
\begin{figure}
    \centering
    \includegraphics[scale=0.55]{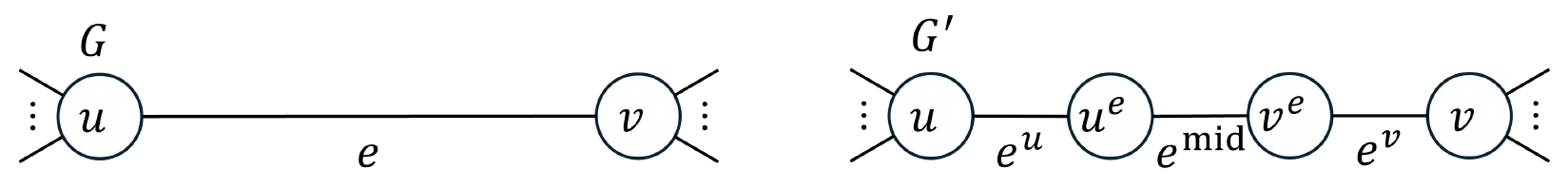}
    \caption{Subdivision of an edge}
    \label{fig:subdivision}
\end{figure}
Let $V^e = \{u^e, v^e\}, E^e_{\textnormal{mid}}=\{e^{\textnormal{mid}}\}, E^e_{\textnormal{end}}= \{e^u, e^v\}$ and $E^e=E^e_{\textnormal{mid}} \cup E^e_{\textnormal{end}}$.
Let $G'$ be the graph obtained by subdividing all edges of $G$.
The vertex set of $G'$ is 
\[
V(G') = V(G) \cup V_{\mathcal{G}} = V(G) \cup \bigcup_{e \in E(G)} V^e,
\]
and the edge set of $G'$ is
\[
E(G') = E_{\textnormal{mid}} \cup E_{\textnormal{end}}
= \bigcup_{e \in E(G)} E^e_{\textnormal{mid}} \cup \bigcup_{e \in E(G)} E^e_{\textnormal{end}}.
\]

\begin{lemma} \label{lem:prop_M}
Let $M'$ be an arbitrary perfect matching of the graph $G'$.
Then, for each $E^e$, exactly one of the following holds: 
\begin{itemize}[\textup{(a)}]
\item[\textup{(a)}] $e^{\textnormal{mid}} \in M'$ and $e^u, e^v \notin M'$
\item[\textup{(b)}] $e^{\textnormal{mid}} \notin M'$ and $e^u, e^v \in M'$
\end{itemize}
\end{lemma}

\begin{proof}
If the edge $e^{\textnormal{mid}}$ is contained in the perfect matching $M'$, then the edges $e^u$ and $e^v$ cannot be contained in $M'$ since both $e^u$ and $e^v$
share one of the endpoints of $e^{\textnormal{mid}}$.
If the edge $e^{\textnormal{mid}}$ is not contained in $M'$, then removing $e^{\textnormal{mid}}$ makes the degrees of vertices $u^e$ and $v^e$ equal to one.
Hence, by the property of a perfect matching, the edges $e^u$ and $e^v$ must be contained in $M'$.
Therefore, Lemma~\ref{lem:prop_M} holds.
\end{proof}

\begin{lemma} \label{lem:one-to-one}
There is a one-to-one correspondence between the perfect matchings of $G$ and those of $G'$.
\end{lemma}

\begin{proof}
For any $M \in \mathcal{M}(G)$, we construct a mapping $\phi$ such that $\phi(M) \in \mathcal{M}(G')$, and show that $\phi$ is bijective.

For each edge $e \in E(G)$, define $\phi(M)$ as follows:
\begin{enumerate}[(a)]
    \item If $e \in M$, then $e^u, e^v \in \phi(M)$.
    \item If $e \notin M$, then $e^{\textnormal{mid}} \in \phi(M)$.
\end{enumerate}
Since $M$ is a perfect matching of $G$, every vertex $v \in V(G)$ is incident to exactly one edge $e \in M$.
By the definition of $\phi$, every vertex $v \in V(G') \setminus V_{\mathcal{G}}$ is incident to exactly one edge $e^v \in \phi(M)$.
Moreover, by Lemma~\ref{lem:prop_M}, each subdivided vertex in $V_{\mathcal{G}}$ is also incident to exactly one edge of $\phi(M)$.
Hence, every vertex of $G'$ is incident to exactly one edge of $\phi(M)$, and therefore $\phi(M)$ is a perfect matching of $G'$.
Figure~\ref{fig:phiM_is_PM} illustrates that every vertex of $G'$ is incident to exactly one edge in $\phi(M)$.
\begin{figure}
    \centering
    \includegraphics[scale=0.45]{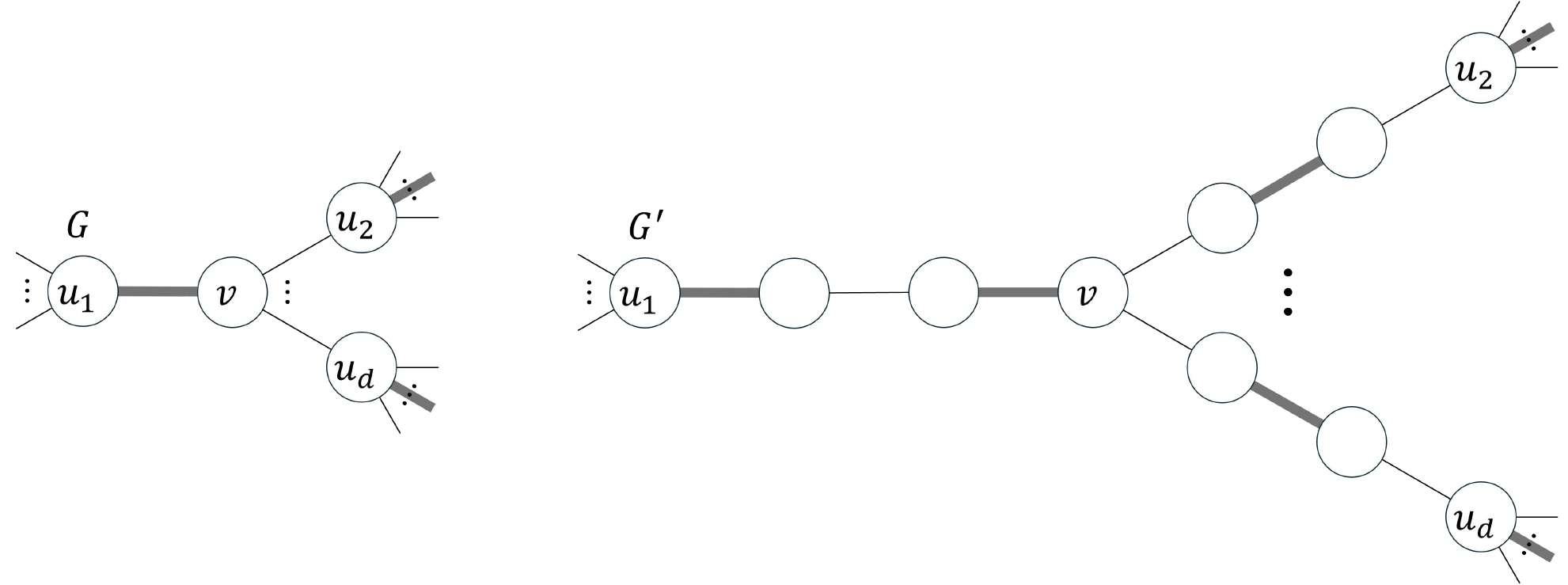}
    \caption{Every vertex of $G'$ is incident to exactly one edge in $\phi(M)$}
    \label{fig:phiM_is_PM}
\end{figure}
Thus, $\phi$ is a mapping from a perfect matching of $G$ to a perfect matching of $G'$.

Let $M' \in \mathcal{M}(G')$.
By Lemma~\ref{lem:prop_M}, for each $E^e$, $M'$ satisfies exactly one of the following:
\begin{enumerate}[(i)]
    \item $e^u, e^v \in M'$ (i.e., $e^{\textnormal{mid}} \notin M'$),
    \item $e^u, e^v \notin M'$ (i.e., $e^{\textnormal{mid}} \in M'$).
\end{enumerate}
Define a mapping $\psi$ as follows:
\begin{enumerate}[(a')]
    \item If $e^{\textnormal{mid}} \in M'$, then $e \notin \psi(M')$.
    \item If $e^{\textnormal{mid}} \notin M'$, then $e \in \psi(M')$.
\end{enumerate}

Suppose that $\psi(M')$ is not a perfect matching.
Then there exists a vertex $v \in V(G)$ such that one of the following holds:
\begin{enumerate}[(1)]
    \item No edge in $\psi(M')$ is incident to $v$.
    \item At least two edges in $\psi(M')$ are incident to $v$.
\end{enumerate}
Let $E_v$ denote the set of edges incident to $v$.
In case (1), by the definition of $\psi$, we have $e^{\textnormal{mid}} \in M'$ for all $e \in E_v$, which implies that no matching edge is incident to $v$ in $G'$, contradicting the fact that $M'$ is a perfect matching.
In case (2), let $e_1, e_2 \in E_v$ be two matching edges.
By the definition of $\psi$, the edges $e_1^v$ and $e_2^v$ are both matching edges incident to $v$ in $G'$, contradicting the matching property.
Hence, $\psi(M')$ is a perfect matching of $G$.

Now consider $\phi(\psi(M'))$.
If $e^{\textnormal{mid}} \in M'$, then by the construction of $\psi$, we have $e \notin \psi(M')$.
The edges $e^u, e^v \not\in\phi(\psi(M'))$ from the construction of $\phi$, and then $e^{\textnormal{mid}} \in \phi(\psi(M'))$ from Lemma~\ref{lem:prop_M}-(a).
Similarly, if $e^{\textnormal{mid}} \notin M'$, then by the construction of $\psi$, we have $e \in \psi(M')$.
The edges $e^u, e^v \in \phi(\psi(M'))$ from the construction of $\phi$, and then $e^{\textnormal{mid}} \notin \phi(\psi(M'))$ from Lemma~\ref{lem:prop_M}-(b).
Therefore, $\phi(\psi(M')) = M'$, and hence $\psi = \phi^{-1}$.
Thus, $\phi$ is bijective.
\end{proof}

\begin{lemma} \label{lem:replace_A}
For any perfect matching $M' \in \mathcal{M}(G')$ and any anti-forcing set $A$ of $M'$, there exists an anti-forcing set $A'$ of $M'$ such that $|A'| \le |A|$ and $A'\subseteq E_{\textnormal{mid}}$.
\end{lemma}

\begin{proof}
If $A \subseteq E_\textnormal{mid}$, the proof is done.
Otherwise, there exists an edge $e \in A \cap E_{\textnormal{end}}$.
This implies that $e \notin M'$ and there exists a vertex $v$, an endpoint of $e$, and an edge $f^v \in E_{\textnormal{end}}$ incident to $v$ such that $f^v \in M'$.
Since $M'$ is a perfect matching, $f^{\textnormal{mid}} \in E^f_{\textnormal{mid}}$ is not in $M'$.
We replace $e$ in $A$ with $f^{\textnormal{mid}}$:
\[
A' = (A \setminus \{e\})\cup \{f^{\textnormal{mid}}\}.
\]
Since $f^{\textnormal{mid}}$ may be included in $A'$, $|A'| \le |A|$ holds.
It is easy to verify that $A'$ forces $f^{\textnormal{mid}} \notin M'$, $f^v \in M'$, and $e \notin M'$.
Therefore, $A$ and $A'$ force the same perfect matching $M'$.
We replace $A$ with $A'$ and repeat this operation until there exists no edge in $A'\cap E_{\textnormal{end}}$.
Finally, we obtain an {\AFS} $A'$ of $M'$ such that $A' \cap E_{\textnormal{end}} = \emptyset$, i.e., $A'\subseteq E_\textnormal{mid}$.
Hence, the lemma holds.
\end{proof}
Figure \ref{fig:Anti_mid} illustrates the replacement of an edge $e$ in an anti-forcing set by the corresponding edge $f^{\textnormal{mid}}$.
\begin{figure}
    \centering
    \includegraphics[scale=0.35]{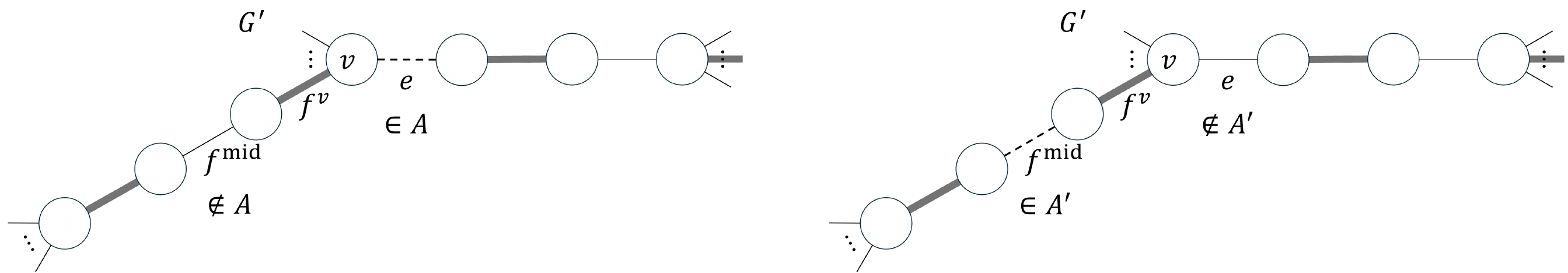}
    \caption{The replacement of an edge}
    \label{fig:Anti_mid}
\end{figure}

\begin{theorem}\label{thm:FGPMtoAFGPM_reduction}
There exists a polynomial-time reduction from {\FGPM} to {\AFGPM}.
\end{theorem}

\begin{proof}
We transform an instance $(G, M, k)$ of {\FGPM} into an instance $(G', \phi(M), k)$ of {\AFGPM}.
The graph $G'$ is constructed from $G$ by subdividing all edges of $G$, and $\phi$ is the mapping defined in Lemma~\ref{lem:one-to-one}.
Since each edge $e \in E(G)$ is replaced with two vertices and three edges, the transformation from $G$ to $G'$ runs in $O(|E(G)|)$ time.

We show that the following statements are equivalent:
\begin{enumerate}[(i)]
    \item There exists a {\FS} $F$ of $M$ in $G$ with $|F| \le k$.
    \item There exists an {\AFS} $A$ of $\phi(M)$ in $G'$ with $|A| \le k$.
\end{enumerate}

\noindent (i)$\Rightarrow$(ii):
For each edge $e \in F$, include the edge $e^{\textnormal{mid}} \in E^e_{\textnormal{mid}}$ in $A$.
We show that the resulting set $A$ is an {\AFS} of $\phi(M)$.
Suppose that $A$ is not an {\AFS} of $\phi(M)$.
Then there exists a perfect matching $M' \neq \phi(M)$ such that $A \subseteq E(G') \setminus M'$.
Since $\phi$ is bijective, we have $\phi^{-1}(M') \neq \phi^{-1}(\phi(M)) = M$ such that $F \subseteq \phi^{-1}(M')$, which contradicts the fact that $F$ is a {\FS} of $M$.
Hence, $A$ is an {\AFS} of $\phi(M)$, and $|A| = |F| \le k$.
Therefore, (i)$\Rightarrow$(ii) holds.\\

\noindent (ii)$\Rightarrow$(i):
Let $A$ be an {\AFS} that uniquely determines a perfect matching $M'$ of $G'$.
By Lemma~\ref{lem:replace_A}, there exists another {\AFS} $A' \subseteq E_{\textnormal{mid}}$ of $M'$ such that $|A'| \le |A|$.
We construct a {\FS} $F$ of the perfect matching $\phi^{-1}(M')$ of $G$ as follows: for each $e^{\textnormal{mid}} \in A'$, include $e$ in $F$.
Suppose that $F$ is not a {\FS} of $\phi^{-1}(M')$.
Then there exists a perfect matching $M'' \neq \phi^{-1}(M')$ such that $F \subseteq M''$.
Since $\phi$ is bijective, we have $\phi(M'') \neq \phi(\phi^{-1}(M')) = M'$ such that $A' \subseteq E(G') \setminus \phi(M'')$, which contradicts the fact that $A'$ is an {\AFS} of $M'$.
Hence, $F$ is a {\FS} of $\phi^{-1}(M')$, and $|F| = |A'| \le |A| \le k$.
Therefore, (ii)$\Rightarrow$(i) holds.
\end{proof}

\begin{theorem} \label{thm:FPMtoAFPM_reduction}
There exists a polynomial-time reduction from {\FPM} to {\AFPM}.
\end{theorem}

\begin{proof}
We transform an instance $(G, k)$ of {\FPM} into an instance $(G', k)$ of {\AFPM}.
The graph $G'$ is constructed from $G$ by subdividing all edges of $G$.
Since the transformation is identical to that in Theorem~\ref{thm:FGPMtoAFGPM_reduction}, it can be done in $O(|E(G)|)$ time.

We show that the following statements are equivalent.
Let $\phi$ be the bijection defined in Lemma~\ref{lem:one-to-one}.
\begin{enumerate}[(i)]
    \item There exists a perfect matching $M$ of $G$ that has a {\FS} of size at most $k$.
    \item There exists a perfect matching $\phi(M)$ of $G'$ that has an {\AFS} of size at most $k$.
\end{enumerate}

The proof proceeds in the same way as that of Theorem~\ref{thm:FGPMtoAFGPM_reduction}, where $M$ is the perfect matching uniquely determined by the {\FS} $F$, and $\phi(M)$ is the perfect matching uniquely determined by the {\AFS} $A$.
\end{proof}

\subsection{{\NP}-completeness of {\sc\bfseries Anti-ForcingGPM} on Bipartite Graphs of Maximum Degree~3}

Let $G'$ be the graph obtained from a graph $G$ by applying the edge subdivision described in Section~\ref{sec:seg}.
Then the following lemma holds.

\begin{lemma}\label{lem:bipartite}
    Let $k \ge 2$.
    If $G$ is a bipartite graph of maximum degree~$k$, then $G'$ is also a bipartite graph of maximum degree~$k$.
\end{lemma}

\begin{proof}
    In the transformation from $G$ to $G'$, each edge $e$ is subdivided by adding only two vertices.
    Therefore, the degrees of the original vertices remain unchanged, and the degrees of the added vertices are at most two.
    Hence, if the maximum degree of $G$ is $k$, then the maximum degree of $G'$ is also $k$.

    Since $G$ is bipartite, there exists a proper $2$-coloring of $G$.
    Thus, for any edge $e=\{u,v\}$, the vertices $u$ and $v$ are assigned different colors.
    For the subdivision of $e$, let $u_e$ and $v_e$ be the two newly added vertices.
    We color $u_e$ with the same color as $v$, and color $v_e$ with the same color as $u$.
    With this coloring, $G'$ admits a proper $2$-coloring.
    Therefore, if $G$ is bipartite, then $G'$ is also bipartite.
\end{proof}

\begin{theorem}\label{AFGPM_NPc}
    {\AFGPM} on bipartite graphs of maximum degree~$3$ is {\NPc}.
\end{theorem}

\begin{proof}
    By Lemma~\ref{lem:AFPM_in_NP}, {\AFGPM} belongs to {\NP}.
    Adams et al.~\cite{adams2004forced} showed that {\FGPM} on bipartite graphs of maximum degree~$3$ is {\NPc}.
    Moreover, by Theorem~\ref{thm:FGPMtoAFGPM_reduction} and Lemma~\ref{lem:bipartite}, there exists a polynomial-time reduction from {\FGPM} on bipartite graphs of maximum degree~$3$ to {\AFGPM} on bipartite graphs of maximum degree~$3$.
    Therefore, {\AFGPM} on bipartite graphs of maximum degree~$3$ is {\NPh}.
    Consequently, the theorem holds.
\end{proof}

\section{{\NP}-completeness of {\sc\bfseries Anti-ForcingPM} on Bipartite Graphs of Maximum Degree~3}
\label{5_3_FPM}

We prove that {\FPM} is {\NPc} on bipartite graphs of maximum degree~$3$.
Combining this result with Theorem~\ref{thm:FPMtoAFPM_reduction} and Lemma~\ref{lem:bipartite}, we prove that {\AFPM} remains {\NPc} even when the input graph is restricted to bipartite graphs of maximum degree~$3$.

First, we transform an instance $(G,k)$ of {\FPM} on a bipartite graph of maximum degree~$5$ into an instance $(G',k)$ of {\FPM} on a bipartite graph of maximum degree~$3$.
For vertices whose degree is at least $4$, we apply the transformations shown in Figures \ref{figure:5_3} and~\ref{figure:4_3}.

\begin{figure}
    \centering
    \includegraphics[scale=0.45]{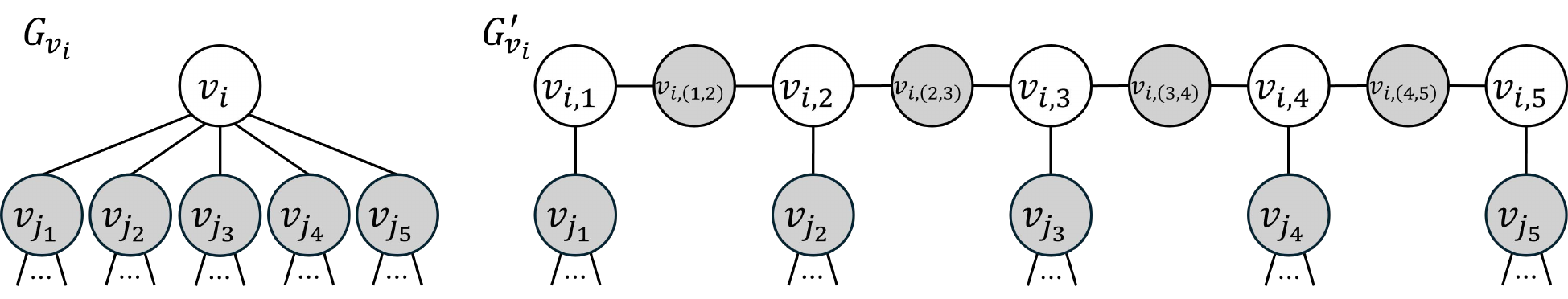}
    \caption{An example of transforming a graph $G_{v_i}$ into a graph $G'_{v_i}$ when the degree of vertex $v_i$ is $5$}
    \label{figure:5_3}
    \centering
    \includegraphics[scale=0.5]{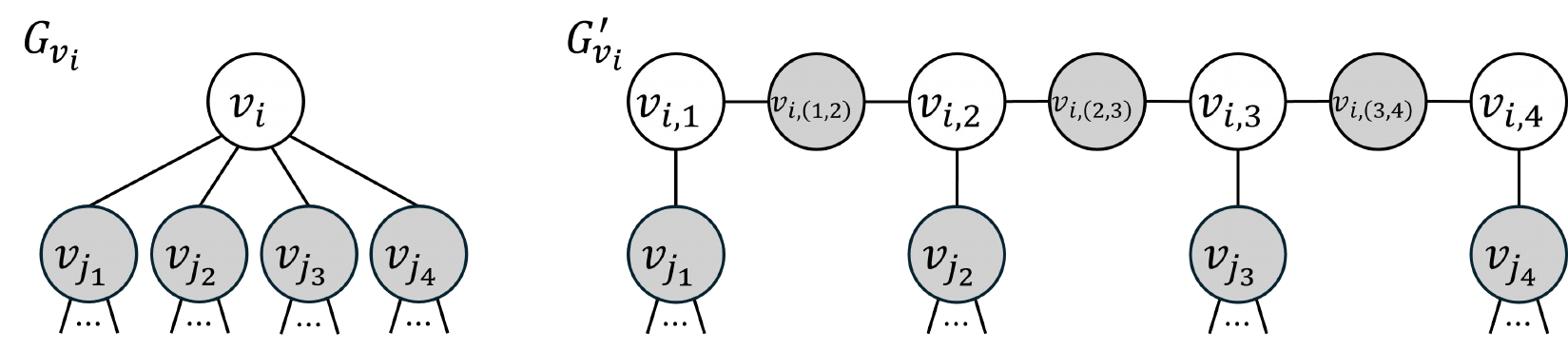}
    \caption{An example of transforming a graph $G_{v_i}$ into a graph $G'_{v_i}$ when the degree of vertex $v_i$ is $4$}
    \label{figure:4_3}
\end{figure}

Let $v_i$ be a vertex of degree~$5$ in graph $G$.
Let $G_{v_i}$ denote the subgraph consisting of the vertex $v_i$, its adjacent vertices $v_{j_s}$ $(1 \leq s \leq 5)$, and the edges $\{v_i, v_{j_s}\}$.
Let $G'_{v_i}$ denote the graph obtained from $G_{v_i}$ by the following transformation.
In $G'_{v_i}$, we introduce vertices $v_{i,s}$ that are copies of $v_i$, and connect each vertex $v_{i,s}$ to $v_{j_s}$ by an edge.
Next, we add vertices $v_{i,(t,t+1)}$ $(1 \leq t \leq 4)$ to $G'_{v_i}$ and connect each vertex $v_{i,(t,t+1)}$ to $v_{i,t}$ and $v_{i,t+1}$ by edges.
Through this process, $G_{v_i}$ is transformed into $G'_{v_i}$ (See Figure \ref{figure:5_3}).
The same transformation is applied to any vertex of degree~$4$ (See Figure \ref{figure:4_3}).

The following lemma holds for the graph $G'$ transformed from the graph $G$.

\begin{lemma} \label{lem:5to3_bipartite}
    If $G$ is a bipartite graph of maximum degree~$5$, then $G'$ is a bipartite graph of maximum degree~$3$.
\end{lemma}

\begin{proof}
    The reduction ensures that the degree of each vertex $v_{i,s}$ in $G'$ is at most three, while the vertices $v_{i,(t,t+1)}$ inserted between them have degree two.
    Therefore, the maximum degree of $G'$ is three.

    Next, we show that $G'$ is bipartite.
    Since $G$ is bipartite, there exists a proper $2$-coloring of $G$.
    In such a coloring, vertex $v_i$ and its adjacent vertices $v_{j_s}$ are assigned different colors.
    By coloring each vertex $v_{i,s}$ in $G'$ with the same color as $v_i$ in $G$, and coloring each vertex $v_{i,(t,t+1)}$ with the same color as $v_{j_s}$ in $G$, it is clear from Figure \ref{figure:5_3} that $G'$ admits a proper $2$-coloring.
    Hence, $G'$ is bipartite.
\end{proof}

The following property holds for perfect matchings of $G'$.

\begin{lemma}\label{lem:prop_M_2}
    Let $M'$ be an arbitrary perfect matching of the graph $G'$.
    Then $M'$ contains exactly one edge of the form $\{v_{i,s}, v_{j_s}\}$ in $G'_{v_i}$.
    Moreover, once such an edge is included, $M'$ is uniquely determined within $G'_{v_i}$.
\end{lemma}
    
\begin{proof}
    Assume that vertex $v_i$ has degree~$5$.
    The case of degree~$4$ follows by the same argument.
    We consider two cases.

    \begin{enumerate}[(1)]
    \item An edge $\{v_{i,s},v_{j_s}\}~(s=1,5)$ is contained in the perfect matching $M'$.
    \item An edge $\{v_{i,s},v_{j_s}\}~(s=2,3,4)$ is contained in the perfect matching $M'$.
    \end{enumerate}
    
    In case (1), suppose that $\{v_{i,1},v_{j_1}\}\in M'$.
    Then vertex $v_{i,(1,2)}$ cannot be matched with $v_{i,1}$, so $\{v_{i,(1,2)}, v_{i,2}\}\in M'$.
    Consequently, $\{v_{i,2}, v_{j_2}\}\notin M'$.
    Repeating the same argument yields
    $\{v_{i,(2,3)}$, $v_{i,3}\}, \{v_{i,(3,4)}, v_{i,4}\}, \{v_{i,(4,5)}, v_{i,5}\}\in M'$,
    and therefore
    $\{v_{i,3}, v_{j_3}\}, \{v_{i,4}, v_{j_4}\}, \{v_{i,5}, v_{j_5}\}\notin M'$.
    The same reasoning applies when $\{v_{i,5},v_{j_5}\}\in M'$.

    In case (2), if $\{v_{i,s}, v_{j_s}\} \in M'$, two adjacent vertices cannot be matched with $v_{i,s}$.
    Since these vertices have degree~$2$, their other incident edges (those not connected to $v_{i,s}$) must belong to $M'$.
    The remainder follows an argument analogous to that of case (1), which proves the lemma.
\end{proof}

\begin{lemma}\label{lem:one_to_one_ch4}
    There is a one-to-one correspondence between the perfect matchings of $G$ and those of $G'$.
\end{lemma}

\begin{proof}
    We show that there exists a bijection $\phi$ that maps each perfect matching $M \in \mathcal{M}(G)$ to a perfect matching $M' \in \mathcal{M}(G')$.
    We consider the correspondence between matchings in $G_{v_i}$ and those in the subgraph $G'_{v_i}$.
    For simplicity, assume that vertex $v_i$ has degree~$5$.
    The case where the degree is four can be handled by constructing an analogous mapping.
    Vertices of degree at most $3$ are not transformed; therefore, for any edge $e \in M$ incident to these vertices, the corresponding edge in $G'$ is included in $\phi(M)$.
    The mapping $\phi$ is obtained by applying a local mapping $\phi_{v_i}$ to each transformed subgraph $G'_{v_i}$, while preserving all matching edges outside the transformed subgraphs.

    For an edge $e=\{v_i, v_{j_s}\}$ in $G_{v_i}$, define the mapping $\phi_{v_i}(M)$ as follows:

    \begin{itemize}
        \item If $\{v_i, v_{j_1}\} \in M$, then $\{v_{i,1}, v_{j_1}\}, \{v_{i,(1,2)}, v_{i,2}\}, \{v_{i,(2,3)}, v_{i,3}\}, \{v_{i,(3,4)}, v_{i,4}\}$, and $ \{v_{i,(4,5)}, v_{i,5}\} \in \phi_{v_i}(M)$.
        
        \item If $\{v_i, v_{j_2}\} \in M$, then $\{v_{i,2}, v_{j_2}\}, \{v_{i,1}, v_{i,(1,2)}\}, \{v_{i,(2,3)}, v_{i,3}\}, \{v_{i,(3,4)}, v_{i,4}\}$, and $\{v_{i,(4,5)}, v_{i,5}\} \in \phi_{v_i}(M)$.
        
        \item If $\{v_i, v_{j_3}\} \in M$, then $\{v_{i,3}, v_{j_3}\}, \{v_{i,1}, v_{i,(1,2)}\}, \{v_{i,2}, v_{i,(2,3)}\}, \{v_{i,(3,4)}, v_{i,4}\}$, and $\{v_{i,(4,5)}, v_{i,5}\} \in \phi_{v_i}(M)$.
        
        \item If $\{v_i, v_{j_4}\} \in M$, then $\{v_{i,4}, v_{j_4}\}, \{v_{i,1}, v_{i,(1,2)}\}, \{v_{i,2}, v_{i,(2,3)}\}, \{v_{i,3}, v_{i,(3,4)}\}$, and $\{v_{i,(4,5)}, v_{i,5}\} \in \phi_{v_i}(M)$.
        
        \item If $\{v_i, v_{j_5}\} \in M$, then $\{v_{i,5}, v_{j_5}\}, \{v_{i,1}, v_{i,(1,2)}\}, \{v_{i,2}, v_{i,(2,3)}\}, \{v_{i,3}, v_{i,(3,4)}\}$, and $\{v_{i,4}, v_{i,(4,5)}\} \in \phi_{v_i}(M)$.
    \end{itemize}
    By Lemma~\ref{lem:prop_M_2}, there is a one-to-one correspondence between the possible matching configurations in $G'_{v_i}$ and the choices of matching edges incident to $v_i$ in $G$.
    Hence, each local mapping $\phi_{v_i}$ is bijective.
    Since $\phi$ is obtained by combining these local bijections and the identity mapping on unchanged edges, $\phi$ is also bijective.
\end{proof}

Next, we show a property of {\FS} on $G'_{v_i}$.

\begin{corollary}\label{cor:forcingset}
    Let $M'$ be a perfect matching of $G'$ and $F_{M'}$ be a forcing set of $M'$.
    Let $M'_{v_i}$ be the partial matching of $G'_{v_i}$ induced by $M'$.
    If $F_{M'_{v_i}} = F_{M'} \cap E(G'_{v_i})$ is nonempty, then there exists a set $F'_{M'_{v_i}} \subseteq E(G'_{v_i})$ such that $|F'_{M'_{v_i}}|=1$ and $(F_{M'} \setminus F_{M'_{v_i}})\cup F'_{M'_{v_i}}$ is also a forcing set of $M'$.
\end{corollary}
 
\begin{proof}
    By Lemma~\ref{lem:prop_M_2}, once the edge $\{v_{i,s}, v_{j_s}\}$ is included in $M'_{v_i}$, the matching $M'_{v_i}$ is uniquely determined.
    Therefore, letting $F'_{M'_{v_i}} = \{\{{v_{i,s}, v_{j_s}}\}\}$ proves the corollary.
\end{proof}

\begin{lemma} \label{FPM53_reduction}
    There exists a polynomial-time reduction from {\FPM} on bipartite graphs of maximum degree~$5$ to {\FPM} on bipartite graphs of maximum degree~$3$.
\end{lemma}

\begin{proof}
    We transform an instance $(G,k)$ of {\FPM} on a bipartite graph of maximum degree~$5$ into an instance $(G',k)$ of {\FPM} on a bipartite graph of maximum degree~$3$.
    The graph $G'$ is obtained by applying the above transformation $G'_{v_i}$ to each vertex $v_i$ of degree at least $4$ in $G$.
    In this transformation, at most eight vertices and eight edges are introduced for each vertex of degree at least $4$ in $G$.
    Therefore, at most $8|V(G)|$ vertices and $8|V(G)|$ edges are added, implying that the transformation from $G$ to $G'$ can be performed in $O(|V(G)|)$ time.

    Next, we show that the following statements are equivalent.
    \begin{enumerate}[(i)]
        \item There exists a perfect matching $M$ in $G$ that has a {\FS} of size at most $k$.
        \item There exists a perfect matching $M'$ in $G'$ that has a {\FS} of size at most $k$.
    \end{enumerate}

    \noindent (i)$\Rightarrow$(ii):
    Let $F_M$ be a {\FS} of a perfect matching $M$ in $G$ satisfying (i).
    By Lemma~\ref{lem:one_to_one_ch4}, there exists a unique perfect matching $M'$ in $G'$ corresponding to $M$.
    For a pair of vertices $v_i$ and $v_j$ of degree at most $3$ in $G$, if an edge $\{v_i, v_j\}\in F_M$, we include $\{v_i, v_j\}$ in $F_{M'}$.
    For each vertex $v_i$ of degree at least $4$, if an edge $\{v_i, v_{j_s}\}\in F_M$, we include $\{v_{i,s}, v_{j_s}\}$ in $F_{M'}$.
    By Corollary~\ref{cor:forcingset}, the set $F_{M'}$ constructed in this way uniquely determines $M'$.\\

    \noindent (ii)$\Rightarrow$(i):
    Let $F_{M'}$ be a {\FS} of a perfect matching $M'$ in $G'$ satisfying (ii).
    By Corollary~\ref{cor:forcingset}, for each subgraph $G'_{v_i}$, if $F_{M'_{v_i}} = F_{M'} \cap E(G'_{v_i})$ is nonempty, then there exists a set $F'_{M'_{v_i}} \subseteq E(G'_{v_i})$ such that $|F'_{M'_{v_i}}|=1$ and $(F_{M'} \setminus F_{M'_{v_i}})\cup F'_{M'_{v_i}}$ is also a forcing set of $M'$.
    Applying this replacement to every subgraph, there exists another {\FS} $F'_{M'}$ of $M'$ such that $|F'_{M'}|\le |F_{M'}|$.
    For each edge $\{v_i, v_j\}\in F'_{M'}$ that is not contained in any subgraph constructed by the reduction, there exists an edge $\{v_i, v_j\}$ in $G$ from the reduction.
    In this case, we include $\{v_i, v_j\}$ in $F_M$.
    For each edge $\{v_{i,s}, v_{j_s}\}\in F'_{M'}$ that is contained in some subgraph constructed by the reduction, there exists an edge $\{v_i, v_{j_s}\}$ in $G$ from the reduction.
    In this case, we include $\{v_i, v_{j_s}\}$ in $F_M$.
    This operation is the inverse of the construction of $F_{M'}$ from $F_M$, so $F_M$ uniquely determines $M$.
\end{proof}

\begin{theorem}\label{FPM_NPc}
    {\FPM} on bipartite graphs of maximum degree~$3$ is {\NPc}.
\end{theorem}

\begin{proof}
    By Lemma~\ref{lem:FPM_in_NP}, {\FPM} belongs to {\NP}.
    {\FPM} on bipartite graphs of maximum degree~$5$ is {\NPc} from the result of Afshani et al.~\cite{afshani2009spectrum}.
    Moreover, by Lemma~\ref{FPM53_reduction}, there exists a polynomial-time reduction from {\FPM} on bipartite graphs of maximum degree~$5$ to {\FPM} on bipartite graphs of maximum degree~$3$.
    Therefore, {\FPM} on bipartite graphs of maximum degree~$3$ is {\NPh}.
    Hence, the theorem holds.
\end{proof}

\begin{theorem}\label{AFPM_NPc}
    {\AFPM} on bipartite graphs of maximum degree~$3$ is {\NPc}.
\end{theorem}

\begin{proof}
    By Lemma~\ref{lem:AFPM_in_NP}, {\AFPM} belongs to {\NP}.
    By Theorem~\ref{FPM_NPc}, {\FPM} on bipartite graphs of maximum degree~$3$ is {\NPc}.
    Furthermore, by Lemma~\ref{lem:bipartite} and Theorem~\ref{thm:FPMtoAFPM_reduction}, there exists a polynomial-time reduction from {\FPM} on bipartite graphs of maximum degree~$3$ to {\AFPM} on bipartite graphs of maximum degree~$3$.
    Therefore, {\AFPM} on bipartite graphs of maximum degree~$3$ is {\NPh}.
    Hence, the theorem holds.
\end{proof}

\bibliographystyle{plainurl}
\bibliography{ref}

@article{adams2004forced,
  title={On the forced matching numbers of bipartite graphs},
  author={Adams, Peter and Mahdian, Mohammad and Mahmoodian, Ebadollah S},
  journal={Discrete Mathematics},
  volume={281},
  number={1-3},
  pages={1--12},
  year={2004},
  publisher={Elsevier},
  doi={10.1016/j.disc.2002.10.002}
}

@article{afshani2009spectrum,
  author       = {Peyman Afshani and
                  Hamed Hatami and
                  Ebadollah S. Mahmoodian},
  title        = {On the spectrum of the forced matching number of graphs},
  journal      = {Australas. {J} Comb.},
  volume       = {30},
  pages        = {147--160},
  year         = {2004}
}

@article{deng2007anti,
  title={The anti-forcing number of hexagonal chains},
  author={Deng, Hanyuan},
  journal={MATCH Commun. Math. Comput. Chem},
  volume={58},
  number={3},
  pages={675--682},
  year={2007}
}

@article{deng2008anti,
  title={The anti-forcing number of double hexagonal chains},
  author={Deng, Hanyuan},
  journal={MATCH Commun. Math. Comput. Chem},
  volume={60},
  number={1},
  pages={183--192},
  year={2008}
}

@article{deng2017anti,
  title={Anti-forcing spectra of perfect matchings of graphs},
  author={Deng, Kai and Zhang, Heping},
  journal={Journal of Combinatorial Optimization},
  volume={33},
  number={2},
  pages={660--680},
  year={2017},
  publisher={Springer},
  doi={10.1007/s10878-015-9986-3}
}

@article{kleinerman2006bounds,
  title={Bounds on the forcing numbers of bipartite graphs},
  author={Kleinerman, Seth},
  journal={Discrete mathematics},
  volume={306},
  number={1},
  pages={66--73},
  year={2006},
  publisher={Elsevier},
  doi={10.1016/j.disc.2005.11.001}
}

@article{LI1997295,
  title = {Hexagonal systems with forcing single edges},
  journal = {Discrete Applied Mathematics},
  volume = {72},
  number = {3},
  pages = {295-301},
  year = {1997},
  issn = {0166-218X},
  author = {Xueliang Li},
  doi={10.1016/0166-218X(95)00116-9}
}

@article{harary1991graphical,
  title={Graphical properties of polyhexes: perfect matching vector and forcing},
  author={Harary, Frank and Klein, Douglas J and {\v{Z}}ivkovi{\v{c}}, Tomislav P},
  journal={Journal of mathematical chemistry},
  volume={6},
  number={1},
  pages={295--306},
  year={1991},
  publisher={Springer},
  doi={10.1007/BF01192587}
}

@article{vukiveevic2007anti,
  title={On the anti-forcing number of benzenoids},
  author={Vuki{\v{e}}evi{\'c}, Damir and Trinajsti{\'c}, Nenad},
  journal={Journal of mathematical chemistry},
  volume={42},
  number={3},
  pages={575--583},
  year={2007},
  publisher={Springer},
  doi={https://doi.org/10.1007/s10910-006-9133-6}
}

@article{randic1985kekule,
  title={Kekule valence structures revisited. {Innate} degrees of freedom of pi-electron couplings},
  author={Randi\'{c}, Milan and Klein, Douglas J},
  journal={Mathematical and computational concepts in chemistry, Wiley, New York},
  pages={274--282},
  year={1985}
}

@article{DemaineKP25,
  author       = {Erik D. Demaine and
                  Kritkorn Karntikoon and
                  Nipun Pitimanaaree},
  title        = {2-Colorable Perfect Matching is {NP}-complete in 2-Connected 3-Regular
                  Planar Graphs},
  journal      = {Theory Comput. Syst.},
  volume       = {69},
  number       = {2},
  pages        = {22},
  year         = {2025}
}

@inproceedings{gabow1999unique,
  title={Unique maximum matching algorithms},
  author={Gabow, Harold N and Kaplan, Haim and Tarjan, Robert E},
  booktitle={Proceedings of the thirty-first annual ACM symposium on Theory of Computing},
  pages={70--78},
  year={1999},
  doi={10.1006/jagm.2001.1167}
}

@article{ZhangHLZ25,
  author       = {Yaxian Zhang and
                  Xin He and
                  Qianqian Liu and
                  Heping Zhang},
  title        = {Forcing, anti-forcing, global forcing and complete forcing on perfect matchings of graphs - {A} survey},
  journal      = {Discret. Appl. Math.},
  volume       = {376},
  pages        = {318--347},
  year         = {2025},
  doi          = {10.1016/J.DAM.2025.06.022},
}

@article{Harary07,
author = {Harary, Frank and  Slany, Wolfgang and  Verbitsky, Oleg},
title = {On the Computational Complexity of the Forcing Chromatic Number},
journal = {SIAM Journal on Computing},
volume = {37},
number = {1},
pages = {1-19},
year = {2007},
doi = {10.1137/050641594},
}

@article{Larson13,
author = {Larson, Craig and Cleemput, Nico Van},
title = {Forcing Independence},
journal = {Croatia Chemica Acta},
year = {2013},
volume = {86},
number = {4},
pages = {469-475},
doi={10.5562/cca2295}
}

@article{Chartrand97,
author={Chartrand, Gary and Gavlas, Heather and Vandell, Robert C. and Harary, Frank},
title={The forcing domination number of a graph},
journal={Journal of Combinatorial Mathematics and Combinatorial Computing},
year={1997},
volume={25},
pages={161-174}
}

@article{Armada19,
title={Forcing Independent Domination Number of a Graph},
author={Cris L. Armada and Canoy, Jr., Sergio R.},
place={Maryland, USA},
DOI={10.29020/nybg.ejpam.v12i4.3484},
journal={European Journal of Pure and Applied Mathematics},
volume={12},
number={4},
year={2019},
month={Oct.},
pages={1371–1381}
}

@inproceedings{GKOS26,
  author       = {Tatsuya Gima and 
                  Yasuaki Kobayashi and 
                  Yota Otachi and 
                  Takumi Sato},
  title        = {Forcing a unique minimum spanning tree and a unique shortest path},
 booktitle    = {Proc. of the 20th International Conference and Workshops on Algorithms and Computation, {WALCOM} 2026},
  series       = {Lecture Notes in Computer Science},
  pages        = {371--385},
  publisher    = {Springer},
  year         = {2026},
  doi          = {10.1007/978-981-95-7127-7\_25}
}

\end{document}